\documentclass[12pt,a4paper]{article}
\usepackage[margin=1.25in]{geometry}
\usepackage{amsfonts, amssymb}
\usepackage{array}
\usepackage{pstricks}
\usepackage{graphicx}
\usepackage{color}
\usepackage{amsthm}
\usepackage{amsmath}
\usepackage{hyperref}
\usepackage{tikz}
\usetikzlibrary{decorations.pathreplacing}
\usepackage{pgfplots}
\usepackage{ dsfont }
\usepackage{graphics}
\usepackage{amsmath}
\usepackage{pifont}
\usepackage{amscd}
\usepackage[latin1]{inputenc}
\usepackage{latexsym}
\usepackage{fancyhdr}
\usepackage{setspace}
\usepackage{enumerate}
\usepackage{hyperref}
\usepackage{bbm}
\usepackage[round]{natbib}
\usepackage[all]{hypcap}  
\usepackage{caption}
\usepackage{subcaption}
\usepackage{placeins}
\usepackage{microtype}

\usetikzlibrary{arrows,calc}

\definecolor{MyRed}{rgb}{0.6,0,0}
\definecolor{MyGreen}{rgb}{0,0.5,0.1}
\definecolor{MyBlue}{rgb}{0,0,0.6}
\definecolor{MyGrey}{rgb}{0.3,0.3,0.4}
\newrgbcolor{Myred}{0.6 0 0}
\newrgbcolor{Mygreen}{0 0.5 0.1}

\hypersetup{colorlinks,
linkcolor=MyGrey,citecolor=MyGrey,filecolor=MyBlue, pdfborder={0 0 0},anchorcolor=MyBlue, urlcolor=MyBlue,pdfstartview=FitH}

\newtheorem{theo}{\color{black}Theorem}
\newtheorem{prop}{\color{black}Proposition}

\newtheorem{lemma}{\color{black}Lemma}

\newcommand{\theoremref}[1]{\hyperref[#1]{Theorem~\ref*{#1}}}
\newcommand{\obsref}[1]{\hyperref[#1]{Observation~\ref*{#1}}}
\newcommand{\lemmaref}[1]{\hyperref[#1]{Lemma~\ref*{#1}}}
\newcommand{\defref}[1]{\hyperref[#1]{Definition~\ref*{#1}}}
\newcommand{\corref}[1]{\hyperref[#1]{Corollary~\ref*{#1}}}
\newcommand{\propref}[1]{\hyperref[#1]{Proposition~\ref*{#1}}}
\newcommand{\assref}[1]{\hyperref[#1]{Assumption~\ref*{#1}}}
\newcommand{\remref}[1]{\hyperref[#1]{Remark~\ref*{#1}}}
\newcommand{\exref}[1]{\hyperref[#1]{Example~\ref*{#1}}}
\newcommand{\appref}[1]{\hyperref[#1]{Appendix~\ref*{#1}}}
\newcommand{\secref}[1]{\hyperref[#1]{Section~\ref*{#1}}}
\newcommand{\subsecref}[1]{\hyperref[#1]{Subsection~\ref*{#1}}}
\newcommand{\figref}[1]{\hyperref[#1]{Figure~\ref*{#1}}}
\newcommand{\footnoteref}[1]{\hyperref[#1]{footnote~\ref*{#1}}}

\newcommand\mydots{\hbox to .8em{.\hss.\hss.}}

\makeatletter
\renewcommand\paragraph{\@startsection{paragraph}{4}{\z@}%
                                    {3.25ex \@plus1ex \@minus.2ex}%
                                    {-1em}%
                                    {\normalfont\normalsize\bfseries\color{MyBlue}}}
\makeatother

\newcommand{\blu}[1]{\textbf{\textcolor{MyBlue}{#1}}}

\pgfplotsset{compat=1.11}

\begin{document}

\onehalfspacing

\title{\blu{\textbf{The Limits of Citation Counts}} 
\thanks{I would like to thank Francis Bloch,  Denis Bouyssou, Ernesto Dal B\'o,  Thierry Marchant, Rafael Treibich and  Ludo Waltman for useful comments on this project, as well as seminar participants at the University of Paris 1 Panth\'eon-Sorbonne and conference participants at Network Science and Economics and Social Choice and Welfare. 
}}

\author{Antonin \textsc{Mac\'{e}}\footnote{CNRS, Paris School of Economics and \'Ecole Normale Sup\'erieure-PSL.}}

\date{September 2023}

\maketitle \setcounter{page}{1}

\begin{abstract}
I study the measurement of scientists' influence using bibliographic data. The main result is an axiomatic characterization of the family of \emph{citation-counting indices}, a broad class of influence measures which includes the renowned $h$-index. The result highlights several limitations of these indices: they are not suitable to compare scientists across different fields, and they cannot account for indirect influence. I explore how these limitations can be overcome by using richer bibliographic information. (JEL: C43, D85)\\

\end{abstract}

\newpage

\section{Introduction}

Bibliographic data are increasingly used to evaluate research activities \citep{hicks2015bibliometrics,hamermesh2018citations}. In the case of individual scientists, citation indices such as the $h$-index \citep{Hirsch} or the Euclidean index \citep{Perry} are often employed to summarize a scholar's research portfolio in a single number. While many alternative citation indices have been developed, their key common characteristic is their parsimonious \emph{informational basis}. Most of these metrics are \emph{citation-counting indices}: they can be computed from the ordered list of citation counts of all papers by a given author. Relatedly, this list of citation counts is prominently displayed on an author's page on Google Scholar, a major institution of modern science routinely used to assess the influence of scientists.

This article aims to uncover the assumptions underlying the reduction of a scholar's bibliographic information to her list of citation counts. Is this informational basis adequate to measure the influence of scientists, and which limitations does it suffer from? To address these questions, I undertake an axiomatic analysis of \emph{influence measures} in a rich framework that encompasses information on authors, their papers and their citation links. The main result is a characterization of the class of \emph{citation-counting indices} by a set of independent properties. By laying down the theoretical foundations for this class of indices, the result helps to frame the discussion about their benefits and limitations, and to explore alternative measures using richer bibliographic data.

The analysis underscores that the substantive restrictions implied by citation-counting indices essentially boil down to two independent axioms. First, Citation Anonymity asserts that a scholar's influence should not depend on the specific identities of the papers citing her own papers. This informational constraint prevents \emph{indirect influence} to be accounted for. This is a severe restriction as indirect influence plays a major role in scientific research due to its cumulative nature \citep{Scotchmer}. For instance, accounting for indirect influence has proved instrumental in identifying influential papers by Nobel Prize-winning physicists \citep{Klosik}.

Second, Author Anonymity precludes a scholar's influence to depend on the authorship of papers written by other scholars. I argue that this axiom is the main culprit behind the recognized inability of citation-counting indices to compare scholars across different scientific fields. To show that, I formalize a property of Field Comparability -- a central yet often implicit desiderata in the bibliometric literature (see below) -- which requires the average influence of scientists to be equal in two independent fields. Alas, this property is essentially incompatible with Author Anonymity. This implies that \emph{citation-counting indices} are inadequate to compare scholars across fields. Finally, I characterize a class of field-comparable indices and explain how simple indices within this class can be constructed.

After discussing the relevant literature in \secref{sLiterature}, the model is introduced in \secref{sModel}. The characterization of citation-counting indices is shown and  discussed in \secref{sCCI}. \secref{sComparison} focuses on the ability of influence measures to compare scientists across fields.  I wrap up with a series of concluding remarks in \secref{sConclusion}.

\subsection{Related literature}\label{sLiterature}

The present article contributes to the literature on the measurement of influence in science -- spurred by the advent of extensive bibliographic databases in the 1960s and 1970s. While initial works addressed the measurement of the influence of academic journals \citep{Garfield, Pinski}, a more recent literature focuses on measuring the influence of scientists, following the introduction of the $h$-index by \citet{Hirsch}. Much of this literature is devoted to the development and analysis of various citation-counting indices, proposed as alternatives to the $h$-index. The axiomatic method has been used to characterize some of these indices, such as the $h$-index \citep{Woeginger, Marchant}, the $g$-index \citep{woeginger2008axiomatic}, the Euclidean index \citep{Perry}, the $\chi$-index \citep{levene2019characterisation}, the class of step-based indices \citep{Chambers} and the class of measure-based indices \citep{de2018ranking}.\footnote{See also \cite{bouyssou2014axiomatic} for a series of characterizations of various citation-counting indices. A distinct approach to evaluate citation indices is to explore how they align with labor market outcomes, see for instance \citet{ellison2013does} for the case of economics.} In all these works, each index (or class of indices) is axiomatized within the class of \emph{citation-counting indices}. The present article thus complements this axiomatic literature by providing a characterization of the very framework in which they are set. In turn, by combining this characterization to that of any index in that literature, one obtains a new characterization of this index, in a richer framework.

\begin{sloppypar}
A notable exception to the previously mentioned literature is \citet{Bouyssou}, who characterize the fractional citation count\footnote{The fractional citation count is very close to be a citation-counting index, in a sense that is made clear in \secref{sCCI}: it satisfies all the axioms that characterize citation-counting indices but the Splitting axiom.} within the broader class of influence measures for authors defined on bibliographic databases. While the present analysis operates in a similar framework, it diverges by characterizing the entire class of citation-counting indices -- encompassing  most indices used in practice -- rather than focusing on a particular index. 
\end{sloppypar}

Several studies emphasize the importance of accounting for indirect influence. \citet{Palacios2004} propose the \emph{invariant method} -- an influence measure for journals where citations carry more weight if they stem from endogenously influential journals. 
 Recursive methods akin to the PageRank algorithm \citep{brin1998anatomy} have also been suggested to measure scientists' influence \citep{Radicchi2009, West}. However, such measures are biased as they only account for indirect influence at the author level rather than at the paper level \citep{Wang}. I discuss in \secref{sConclusion} how to construct a measure accounting for indirect influence that does not suffer from this bias.


\begin{sloppypar}
As common influence measures vary significantly across fields, a pervasive theme of the bibliometric literature is the idea of field-normalization \citep{Ioannidis,waltman2019field}. For instance, \citet{Radicchi2008} argue that appropriately normalizing citations across diverse scientific fields yields a \emph{universal} distribution of normalized citations per paper. Building on this work, \citet{Perry} introduce the \emph{homogeneity} property for citation-counting indices, ensuring that re-scaling citations across fields preserves the relative ranking of scientists within a given field. Overall, field-normalization emerges as a consensual recommendation from the bibliometric community \citep{hicks2015bibliometrics}. With regard to this literature, the present article seeks to clarify the notion of comparability across fields and  to assess when and why citation indices may be comparable (see \secref{sComparison} for a detailed discussion).
\end{sloppypar}

\section{Model}\label{sModel}

\blu{Authors and papers}. The set of authors is denoted by $\mathcal{A}$, and the number of authors is $A$. Each author $a\in \mathcal{A}$ has written a set of papers $\mathcal{P}_a$, containing $P_a$ papers. The set of all papers is denoted by $\mathcal{P}=\cup_{a\in \mathcal{A}}\mathcal{P}_a$. The collection of the sets of papers written by each author is denoted by $\mathcal{P}_{\mathcal{A}}=(\mathcal{P}_a)_{a\in \mathcal{A}}$.

\noindent \blu{Citations and references}. Papers in $\mathcal{P}$ are related through a directed network $n: \mathcal{P}\times\mathcal{P} \rightarrow \{0,1\}$, where $n(p,q)=1$ means that $q$ cites $p$, a relation which we will often interpret as an influence of $p$ on $q$. For a given paper $q\in \mathcal{P}$, the set of references is $\mathcal{R}_q=\{p\in \mathcal{P}, n(p,q)=1\}$, and its number of references is $R_q$. For a given paper $p\in \mathcal{P}$, the set of citations is $\mathcal{C}_p=\{q\in \mathcal{P}, n(p,q)=1\}$, and its number of citations is $C_p$.

\noindent \blu{Bibliographic databases}. A database is a collection $d=(\mathcal{A}, \mathcal{P}_{\mathcal{A}},n)$. For the clarity of the exposition, I focus in this article on the set of databases $\mathbb{D}$ such that each paper is single-authored, there is no self-citation, and each author cites at least one other author. Formally:
$$\mathbb{D}=\left\{d=(\mathcal{A}, \mathcal{P}_{\mathcal{A}},n) \quad, \quad 
\begin{array}{ll}
\forall a,b\in\mathcal{A}, a\neq b, & \mathcal{P}_a\cap \mathcal{P}_b=\emptyset \\
\forall a\in \mathcal{A},\forall p,q\in \mathcal{P}_a, & n(p,q)=0 \\
\forall a\in \mathcal{A}, \exists q\in \mathcal{P}_a, \exists p\in \mathcal{P} \backslash \mathcal{P}_a, & n(p,q)=1 
\end{array}
\right\}.
$$
These assumptions are discussed in \secref{sConclusion}. 

\noindent \blu{Influence measures}. An \emph{influence measure (or index) for authors} $f$ assigns to any database $d=(\mathcal{A}, \mathcal{P}_{\mathcal{A}},n)\in \mathbb{D}$ a vector of scores $f(d)=(f_a(d))_{a\in \mathcal{A}}\in (\mathbb{R_+})^{\mathcal{A}}$. The number $f_a(d)$ is the influence of author $a$ in the database $d$, as measured by the index $f$.\footnote{Note that we consider here a \emph{cardinal} measure for the simplicity of exposition. A characterization result similar to \theoremref{tMain} can be obtained for an \emph{ordinal} measure (assigning an ordering on $\mathcal{A}$ to each database $d$), by adapting each axiom to the ordinal setting (proof available upon request). }

\noindent \blu{Neutral measures}. A measure $f$ is \emph{neutral} if the allocation of influence to authors is independent of their names and of the names of the papers. Formally, neutrality requires that for any bijection of authors $\pi :\mathcal{A}\rightarrow \mathcal{A}'$, for any bijection of papers $\sigma: \mathcal{P}\rightarrow \mathcal{P}'$, if $d=(\mathcal{A}, \mathcal{P}_{\mathcal{A}},n)$, if $d'=(\mathcal{A}', \mathcal{P}'_{\mathcal{A}},n')$ with $\mathcal{P}'_{\mathcal{A}}=\left(\sigma(\mathcal{P}_{\pi^{-1}(a)}) \right)_{a\in \mathcal{A}'}$ and $n'\left(\sigma(p),\sigma(q) \right)=n(p,q)$ for all $p,q\in \mathcal{P}$, then $f_{\pi(a)}(d')=f_a(d)$ for all $a\in \mathcal{A}$. Neutrality is a natural requirement that will be satisfied by all measures considered in the article.

\noindent \textbf{Example 1: the Euclidean index}. This index, introduced  by \citet{Perry}, is defined by:
$$f_a(d)=\left[\sum_{p\in \mathcal{P}_a}\left(\sum_{q\in \mathcal{P}} n(p,q)\right)^2\right]^{1/2} = \left[\sum_{p\in \mathcal{P}_a}(C_p)^2\right]^{1/2} .$$
This index belongs to the class of \emph{citation-counting indices} (see below) as the influence of an author $a$ only depends on the list of her citation counts $(C_p)_{p\in \mathcal{P}_a}$. As this index is homogeneous (of degree $1$), it is not comparable across fields. To see this point, consider a database split into two distinct fields with the same number of authors. Assume that for every author in the second field, there is a corresponding author in the first field with twice as many citations for each paper. The average index in the first field will then be twice as high as in the second field, while Field Comparability requires these averages to be equal  (see discussion on \secref{sComparison}). Furthermore, the Euclidean index neglects indirect influence, as it exclusively relies on direct citations.

\noindent \textbf{Example 2: a comparable measure (\emph{intellectual debt}):} 
$$f_a(d)=\sum_{p\in \mathcal{P}_a} \sum_{b\in \mathcal{A}}\dfrac{1}{P_b} \sum_{q\in \mathcal{P}_b} \dfrac{1}{R_q}n(p,q).$$
The measure can be interpreted as follows: (i) each author holds one unit of intellectual debt to the broader scientific community, (ii) this debt is then evenly distributed among her papers, and within each paper, it is split equally among its references, (iii) the score of an author coincides with the total credit she owns in the scientific community.

Observe that this index is comparable across fields. Indeed, it satisfies the following accounting equation:
\begin{equation} \label{eAccounting}
\sum_{a\in \mathcal{A}}f_a(d)= A. 
\end{equation}
Essentially, the cumulative score in the database, by its very definition, represents the total intellectual debt, which equates to the number of authors.\footnote{Formally: $\displaystyle{\sum_{a\in \mathcal{A}} f_a(d)=\sum_{p\in \mathcal{P}} \sum_{b\in \mathcal{A}}\frac{1}{P_b}\sum_{q\in \mathcal{P}_b}\dfrac{1}{R_q}n(p,q)=\sum_{b\in \mathcal{A}} \frac{1}{P_b}\sum_{q\in \mathcal{P}_b} \left(\sum_{p\in \mathcal{P}}\dfrac{1}{R_q}n(p,q)\right)=A.}$} Therefore, if the database is divided in two disjoint fields, the average influence will be the same in the two fields.   As the Euclidean index, this measure is a simple index that only depends on direct citations (see \secref{sConclusion} for an extension of the measure accounting for indirect influence).

\noindent \textbf{Example 3: a comprehensive measure:}
$$f_a(d)=\sum_{p\in \mathcal{P}_a} \sum_{q\in \mathcal{P}} \sum_{r\in \mathcal{P}} n(p,q)n(q,r).$$
As the Euclidean index, this measure is a simple index that is not comparable across fields. However, the index accounts for some indirect influence, as $f_a(d)$ represents the cumulative citations garnered by papers citing $a$'s papers.


\section{Citations-counting indices}\label{sCCI}

Following the introduction of the $h$-index by \citet{Hirsch}, numerous other indices that ``count citations'' have been proposed. The popularity of these measures likely stems from their parsimonious use of information: an author is identified with the collection of the citation counts of each of her papers. This collection can be viewed as a multiset on $\mathbb{N}$, i.e. a function $m_a[d]:\mathbb{N}\rightarrow \mathbb{N}$, defined by $\forall k\in \mathbb{N},\ m_a[d](k)=\#\{p\in \mathcal{P}_a, C_p=k\}$. The number $m_a[d](k)$ represents the number of papers written by $a$ that have received exactly $k$ citations in the database $d$. The set of all multisets on $\mathbb{N}$ is denoted by $\mathcal{M}$. 

\noindent \blu{Citation-counting indices}. An influence measure $f$ is a \emph{citation-counting index} if there exists a function $F: \mathcal{M}\rightarrow \mathbb{R}$ such that: $\forall d\in \mathbb{D}, \quad f_a(d)=F\left(m_a[d]\right).$

Citation-counting indices include the $h$-index, the Euclidean index, and all other indices mentioned in \secref{sLiterature}, with the exception of the fractional citation count.

\subsection{Axioms}

This section introduces five axioms on influence measures. For the first axiom, we consider two databases $d,d'\in \mathbb{D}$ that do not overlap, i.e. such that $\mathcal{A}\cap\mathcal{A}'=\emptyset$ and $\mathcal{P}\cap\mathcal{P}'=\emptyset$. In that case, we note $d'' = d \oplus d'$ the union of the two disjoint databases, defined by $d''=(\mathcal{A}'', \mathcal{P}_{\mathcal{A}''}'',n'')$, where  $\mathcal{A}''=\mathcal{A}\cup \mathcal{A}'$, $\mathcal{P}_{\mathcal{A}}''=\left( (\mathcal{P}_a)_{a\in \mathcal{A} },(\mathcal{P}'_a)_{a\in \mathcal{A}' } \right)$ and $n''=n\mathds{1}_{\mathcal{P}\times \mathcal{P}} + n'\mathds{1}_{\mathcal{P}'\times \mathcal{P}'}$.\footnote{Note that with this definition, $d''\in \mathbb{D}$. Note also that we use the notation $\mathds{1}_{\mathcal{P}^0\times\mathcal{P}^1}$ to denote the network  $n$ such that $n(p,q)=1$ if and only if $(p,q)\in \mathcal{P}^0\times\mathcal{P}^1$.} The axiom requires that the addition of an unrelated database leaves the score of any author unchanged.

\noindent \blu{Separability}. For any two disjoint databases $d,d'\in \mathbb{D}$, we have:
 $$\forall a\in \mathcal{A},\quad f_a(d\oplus d')=f_a(d).$$
The following property requires that the references of an author do not affect her score. Note that, while this property is desirable for a fixed database, it also prevents an author from increasing her score by manipulating her references.\footnote{In the context of the measurement of influence for journals, \citet{Koczy} observe that the invariant method is vulnerable to precisely  this type of manipulation.}  

\noindent \blu{Reference Independence}. For any author $a\in \mathcal{A}$, we have:
$$\forall q\in \mathcal{P}_a, \forall p\in \mathcal{P} \backslash \mathcal{P}_a, \quad n(p,q)=0 \quad \Rightarrow \quad f_a(d'=(\mathcal{A}, \mathcal{P}_{\mathcal{A}},n+\mathds{1}_{\{(p,q)\}})) = f_a(d).$$
The next axiom states that splitting an uncited paper into two papers with disjoint reference lists should leave the score of any author unaffected. This property limits the possibility of normalizing the source of citations.
  
\noindent \blu{Splitting}. Let $b\in \mathcal{A}$ and $q\in \mathcal{P}_b$ such that $C_q=0$. If $\mathcal{P}_b'=\mathcal{P}_b\cup \{q'\}$, with $q'\notin \mathcal{P}$, and $n'$ is such that $\mathcal{R}'_q \cup \mathcal{R}'_{q'}=\mathcal{R}_q$, $\mathcal{R}'_q \cap \mathcal{R}'_{q'}=\emptyset$ and $\forall r\neq q,   \mathcal{R}'_r=\mathcal{R}_r$ and $\forall a \neq b, \mathcal{P}_a'=\mathcal{P}_a$, then:
$$\forall a\in \mathcal{A},\qquad f_a(d'=(\mathcal{A}, \mathcal{P}_{\mathcal{A}}',n'))=f_a(d).$$
The following axiom requires the score of an author to be independent of the specific identity of its citations. As we shall discuss in \secref{sResult}, this property prevents the measure from capturing indirect influence at the level of each paper.

\noindent \blu{Citation Anonymity}. For any permutation $\sigma:\mathcal{P}\rightarrow \mathcal{P}$ such that $\forall a \in \mathcal{A}, \sigma(\mathcal{P}_a)=\mathcal{P}_a$, if $\forall p,q\in \mathcal{P}, n'(p,q)=n(p,\sigma(q))$, then:
$$\forall a\in \mathcal{A},\qquad f_a(d'=(\mathcal{A}, \mathcal{P}_{\mathcal{A}},n'))=f_a(d).$$
The last axiom of this section requires the score of an author $a$ to be independent of the authorship of papers not written by $a$. As we shall discuss in \secref{sComparison}, this property essentially prevents the measure to be comparable across fields.

\noindent \blu{Author Anonymity}. For any two databases $d=(\mathcal{A}, \mathcal{P}_{\mathcal{A}},n)$ and $d'=(\mathcal{A}, \mathcal{P}_{\mathcal{A}}',n)$ in $\mathbb{D}$ such that $\mathcal{P}=\mathcal{P}'$, for any $a$ such that $\mathcal{P}_a=\mathcal{P}'_a$, we have:
$$f_a\left(d'=(\mathcal{A}, \mathcal{P}_{\mathcal{A}}',n)\right)=f_a(d).$$

\subsection{Result and discussion}\label{sResult} 
The main result of this section is the following characterization of citation-counting indices.

\begin{theo}\label{tMain}
A neutral influence measure $f$ satisfies Separability, Reference Independence, Splitting, Citation Anonymity and Author Anonymity if and only if $f$ is a citation-counting index. Moreover, the 5 axioms are independent.
\end{theo}

\theoremref{tMain} clarifies the assumptions underlying the use of citation-counting indices, such as the $h$-index or the Euclidean index, to measure the influence of authors in a database. On the one hand, such an index satisfies two properties that seem particularly appealing. Separability means that an author's score does not depend on  bibliographic information on a field distinct from the author's field. Reference Independence implies that an author's references do not contribute to the assessment of her own influence. On the other hand, citation-counting indices combine three other independent properties that might be too restrictive for measuring the influence of scientists.

The first restrictive assumption, and perhaps the more benign one, is Splitting. Because splitting an uncited paper does not alter any author's index, this property suggests that papers with more references will be relatively more important when attributing credit to authors. For instance, a paper with 50 references would be equivalent to 10 papers with 5 references each. This contradicts the idea of source-normalization \citep{waltman2019field}, according to which each source of citations must account for the same level of influence. As an example, Splitting is violated by the fractional citation count \citep{Bouyssou}.\footnote{Fractional citation count: $\displaystyle{f_a(d)=\sum_{p\in \mathcal{P}_a}\sum_{q\in \mathcal{P}}\dfrac{1}{R_q}n(p,q)}$.\label{fnFCC}} 

The second restrictive axiom is Citation Anonymity. When it holds, indirect influence between papers cannot be taken into account by the index. Consider for instance a paper $p$, cited by a paper $q$, which is itself cited by a paper $r$. By Citation Anonymity, one can permute the papers citing $p$ without altering the index. It follows that the statement ``$p$ is indirectly cited by $r$'' is not a meaningful statement for the computation of the index.  This limitation is important, as it implies that a citation from an uncited paper holds the same weight as one from a paper with 100 citations.

Note that, in all rigor, Citation Anonymity still permits to account for some indirect influence, at the author level. Indeed, the permutation of papers used to define the axiom respects the partition of papers into authors, and therefore preserves the global flow of citations between authors. This allows to use network-based methods akin to the PageRank algorithm to measure the influence of authors recursively \citep{Radicchi2009, West}. However, these procedures inevitably lead to biases in the assessment of indirect influence. Consider for instance an author $a$, with two papers, one that is uncited and one that is highly cited. If author $b$ gets a citation from $a$, she will gain a lot of influence in a recursive index (since $a$ is a highly cited author), even if she is cited by the uncited paper. Assessing indirect influence at the author level therefore seems too crude if one has access to the citation network between papers \citep{Wang}.

The last restrictive property is Author Anonymity. It requires that the authorship of  papers not written by $a$ does not affect the influence of $a$. As an illustration, a paper with 100 citations will be judged equally influential whether these citations all come from the same author or from 100 different authors. This seems disputable from a normative perspective, and also raises the issue of the possibility of manipulating the measure.\footnote{Here, I refer to a potential manipulation by an author who could add papers at no (or a very small) cost. The remark implicitly assumes that papers are easier to manipulate (to counterfeit) than authors.} But more importantly, this single axiom prevents the influence measure to be comparable across fields, as discussed in the next section.

\section{Comparisons of authors across fields}\label{sComparison}


Common citation indices are known to produce significant differences across fields, 
because of their varying sizes and of their different traditions in terms of publications and citations. As a result, there seems to be a consensus that citations should be normalized by field, so as to allow for meaningful comparisons across fields \citep{Ioannidis,waltman2019field}. For instance, \citet{Perry} propose to divide the number of citations of a paper by the average number of citations per paper in the paper's field, before computing the Euclidean index.\footnote{It has also been suggested to normalize a given index by field \citep{Kaur}. \citet{Perry} provide a discussion of the relative merits of the two approaches. See \cite{waltman2019field} for a comprehensive survey of this literature.}  In this section, I argue that field-normalization leads to serious conceptual difficulties. I thus propose a weaker notion of Field Comparability. I show that this property is incompatible with Author Anonymity, and is therefore violated by any citation-counting index, under mild conditions. I conclude the section by a characterization of field-comparable indices.

There are at least three important issues raised by the procedure of field-normalization. First, field-normalization \emph{creates biases} between interrelated fields. Indeed, even if most citations are issued within fields, there are important flows of citations across fields, and these flows need not be balanced. For instance, \citet{angrist2020inside} report substantial and unbalanced flows of citations in recent decades across fields in the social sciences.\footnote{In Figure 2 of \citet{angrist2020inside}, one observes that, in 2010, less than $3\% $ of citations from economics journals are given to political science journals, while almost $15\% $ of citations from political science journals are given to economics journals. Note that, in this figure, citations from a given journal are weighted by the importance of the journal within the journal's field.} In such circumstances, field-normalization underweighs the influence of influential fields, and overweighs the influence of fields ``under influence''.

Second, field-normalization is sensitive to the level of aggregation retained to perform the normalization. Consider two authors who have the same list of citations counts. Suppose that the first is an economic theorist, while the second is a graph theorist. As papers are more cited in economics than in mathematics, it seems that the second author should be considered more influential, if one adheres to the field-normalization paradigm. At the same time, if one observes that papers are more cited in graph theory than in economic theory, one must also conclude that the first author is more influential than the second one. It follows that field-normalization is an ambiguous notion, that crucially depends on the level at which the normalization is performed \citep{Zitt2005}.

A third issue relates to the difficulty of reaching a consensual field classification system \citep{waltman2019field}. Different methodologies, whether algorithmic or expert-driven, can lead to distinct field classification outcomes. The arbitrariness of the retained classification system is most acute when a scholar lies at the intersection of two fields, say, one with a high citation rate and one with a low citation rate. The measured influence of such scholar is then highly dependent of the field assigned to her by the classification system.


I have argued here that, while the \emph{objective} of field-normalization is consensual, the \emph{procedure} of field-normalization raises three important issues that limit the relevance of its outcomes. To distinguish between the two, I now formalize a minimal requirement (on influence measures) that captures the \emph{objective} of field-normalization, without imposing restrictions when the \emph{procedure} of field-normalization is problematic. The following axiom requires the average influence to be equal for two different fields of science, \emph{only in the extreme case in which the two fields are completely disjoint}.\footnote{A similar idea was developed by \citet{Waltman2013}, who observed that the SNIP index - an impact indicator for journals - will have the same average for two disjoint fields of science, under a couple of benign conditions. Relatedly, a slightly stronger property - \emph{insensitivity to field differences} - has been proposed for journals impact indicators in \citet{Waltman2010}.}


\noindent \blu{Field Comparability}. If $d\in \mathbb{D}$ is divided in two disjoint fields, i.e. there is a partition $\mathcal{A}=\mathcal{A}^1\cup \mathcal{A}^2$ such that for all $(p,q)\in (\cup_{a\in \mathcal{A}^1 }\mathcal{P}_a) \times (\cup_{a\in \mathcal{A}^2 }\mathcal{P}_a) $, $n(p,q)=n(q,p)=0$. Then:
$$\dfrac{1}{A^1}\sum_{a\in \mathcal{A}^1}f_a(d)=\dfrac{1}{A^2}\sum_{a\in \mathcal{A}^2}f_a(d).$$
Note that an influence measure can satisfy the property, while still accounting for global citation flows between fields. This addresses the first issue. Furthermore, the definition of fields in this axiom is unambiguous, this addresses the second and third issues.

In the sequel, I prove the incompatibility between this property and Author Anonymity under mild conditions. We say that  $f$ is non-degenerate if there is a database $d$ and an author $a$ such that $f_a(d)>0$, and that $f$ is separable if it satisfies Separability. The following axiom further requires that authors with no citation have no influence.

\noindent \blu{Null Author}.\footnote{The terminology follows the \emph{Null Player} property from the theory of cooperative games \citep{Shapley}.} Let $d\in \mathbb{D}$ and $a\in \mathcal{A}$. We have $  \forall p\in \mathcal{P}_a, \ C_p=0 \ \Rightarrow \ f_a(d)=0$.

All common citation indices are non-degenerate, separable and satisfy this property. We obtain the following impossibility result.

\begin{prop}\label{pImpossibility}
There is no non-degenerate and separable influence measure satisfying Null Author, Author Anonymity and Field Comparability.
\end{prop}

As a corollary of \propref{pImpossibility} and \theoremref{tMain}, we obtain that citation-counting indices cannot be field-comparable, unless they violate benign conditions. While field-normalization can still be applied before or after index computation, we have argued that this procedure is inherently arbitrary and potentially biased.

We conclude the section by a characterization of the class of (separable and neutral) influence measures that are field-comparable.

\begin{prop}\label{pPossibility}
A neutral and separable influence measure $f$ satisfies Field Comparability if and only if it satisfies the accounting equation \eqref{eAccounting} up to a multiplicative constant.
\end{prop}

\propref{pPossibility} highlights that, for a reasonable measure to be field-comparable, the total influence in any given field must be normalized, i.e. proportional to its number of authors. This observation resonates with the source-normalization paradigm in bibliometrics \citep{waltman2019field}, which recommends to normalize citations originating from the same source (at a micro level), for instance weighing a citation by the number of references of the citing paper. Source-normalization is often thought to improve comparability across fields: if for instance two fields only differ by the number of references per paper, the procedure yields comparable results across the two fields. \propref{pPossibility} complements this literature by highlighting that Field Comparability can in fact be universally achieved (whatever differences across fields may exist),  by requiring a macro condition -- the accounting equation \eqref{eAccounting} -- on the influence measure. Finally, Example 2 illustrates that this macro condition can be satisfied by imposing a micro condition: that the total influence exerted on any given author is constant. This micro condition is similar in spirit to source-normalization, but it takes place at the level of authors rather than individual papers.

\section{Concluding Remarks}\label{sConclusion}

This study delineates the assumptions made when  using citation-counting indices to measure scientists' influence from bibliographic data. As these assumptions are quite restrictive,  I now discuss possible ways to leverage bibliographic data to address these limitations.
\medskip

\noindent \textbf{Intellectual debt and indirect influence.} We have seen above that the intellectual debt measure from Example 2 satisfies Field Comparability. We explore how to adapt the measure to account for indirect influence, while preserving Field Comparability. The measure from Example 2 can be re-written as follows:
\begin{equation} \label{eID}
f_a(d)  = \sum_{b\neq a } \dfrac{\sum_{p\in\mathcal{P}_a}\sum_{q\in \mathcal{P}_b}g(p,q) }{\sum_{p\in \mathcal{P}\backslash\mathcal{P}_b}\sum_{q\in \mathcal{P}_b}g(p,q)},
\end{equation}
where for any $p,q\in \mathcal{P}$, the influence of $p$ on $q$ is defined by $g(p,q)=g_1(p,q):=\frac{1}{R_q}n(p,q)$. Let $\delta\in(0,1)$ be a discount factor reflecting the relative importance of indirect versus direct citations. Consider now the overall influence of $p$ on $q$ defined by $g^\delta(p,q)=(1-\delta)\sum_{k=1}^{+\infty}\delta^{k-1}g_k(p,q)$, where $g_k$ is defined inductively by $ g_{k+1}(p,q)=\sum_{r\in\mathcal{P}}g_{k}(p,r)g_1(r,q)$. The measure $g^\delta(p,q)$ aggregates the influence of $p$ and $q$ inferred from both direct and indirect citations.\footnote{\citet{Klosik} similarly suggest to aggregate indirect influences order of different orders through a discounted sum, for the measurement of the impact of a single paper.} By using $g=g^\delta$ in \eqref{eID}, we obtain an influence measure $f^\delta$ which accounts for indirect influence, and also satisfies Field Comparability, as the accounting equation \eqref{eAccounting} is preserved.\footnote{Note that $f^\delta$ does not satisfy Reference Independence, but it can be easily adapted to do so. For a detailed discussion of the construction of $f^\delta$, see a previous version of the paper \citep{mace2017measuring}, where each step of the construction was axiomatized.}

\medskip

\noindent\textbf{Real databases.} In order to ease the exposition of the paper, I introduced a stylized setting with single-authored papers and no self-citations. I briefly discuss these assumptions here.

In modern research, a growing number of papers are written by a group of authors, rather than by a single author \citep{Wuchty,Hamermesh}. 
It is therefore important to adapt influence measures to databases with this feature, and this can be done flexibly. Suppose that for each paper $p\in \mathcal{P}$, we have a distribution of weights $\omega_p=(\omega_p^a)_{a \in \mathcal{A}}$, reflecting the various contributions of authors in $\mathcal{A}$ to $p$, and such that $\sum_{a\in \mathcal{A}}\omega_p^a=1$. For instance, for each co-author $a$ of $p$, $\omega_p^a$ could be defined as the inverse of the number of co-authors, while $\omega_p^b=0$ for all other authors $b\in \mathcal{A}$ \citep{Radicchi2009}.\footnote{Note that alternative weights depending on the ordering of co-authors may be more appropriate in the natural sciences. Yet another possibility would be to apply an endogenous sharing rule, for which the relative share of a co-author on a paper depends on her overall influence in the database \citep{flores2020teamwork}.} Then, the intellectual debt measure $f$ can be adjusted by multiplying each term of the form $g(p,q)$ with $(p,q)\in \mathcal{P}_a\times  \mathcal{P}_b$ by $\omega_p^a\omega_q^b$.  This adjustment preserves the property of Field Comparability.



The model also neglected self-citations, which are pervasive in real bibliographic databases. This issue can be dealt with easily by erasing any citation for which the same author appears both in the cited and citing articles. Finally, I have assumed that a database contains an accurate register of authors for each paper. In practice, authorship could be difficult to retrieve for some articles. One way to deal with this issue is to restrict the database to papers written by authors from a trusted resister, such as ORCID, or RePEc in economics.

\medskip

\noindent\textbf{Going forward.} The analysis has been confined to the measurement of the influence at the level of individual scholars. This setting is empirically relevant as bibliographic data are often used to compute individual measures, although these data are arguably too sparse to actually reflect the quality of scientific works. Influence measures are presumably more meaningful for larger entities such as journals, departments, universities or geographic areas.\footnote{See for instance \cite{checchi2021have}. Yet, using such indicators for public policy remains vulnerable to manipulation, see for instance \cite{seeber2019self}.} The limitations of citation-counting indices extend to these settings as well, and the alternative measures discussed in the paper could be adapted for such endeavors.

Another promising avenue for further research concerns the vulnerability of influence measures to potential manipulations by a scientist or a group of scientists. Consider for instance a group of scientists that can manufacture papers and citations between those papers at no (or very little) cost. Such a group could reach any influence level under a citation-counting index, but its leverage would be more limited under a field-comparable measure, such as the intellectual debt measure, as its overall impact is bounded by the size of the group.

\newpage

\singlespacing

\bibliographystyle{ecta}\bibliography{Influence}

\newcommand{\noop}[1]{}
\begin{thebibliography}{38}
\newcommand{\enquote}[1]{``#1''}
\expandafter\ifx\csname natexlab\endcsname\relax\def\natexlab#1{#1}\fi

\bibitem[\protect\citeauthoryear{Angrist, Azoulay, Ellison, Hill, and
  Lu}{Angrist et~al.}{2020}]{angrist2020inside}
\textsc{Angrist, J., P.~Azoulay, G.~Ellison, R.~Hill, and S.~F. Lu} (2020):
  \enquote{Inside job or deep impact? Extramural citations and the influence of
  economic scholarship,} \emph{Journal of Economic Literature}, 58, 3--52.

\bibitem[\protect\citeauthoryear{Bouyssou and Marchant}{Bouyssou and
  Marchant}{2014}]{bouyssou2014axiomatic}
\textsc{Bouyssou, D. and T.~Marchant} (2014): \enquote{An axiomatic approach to
  bibliometric rankings and indices,} \emph{Journal of Informetrics}, 8,
  449--477.

\bibitem[\protect\citeauthoryear{Bouyssou and Marchant}{Bouyssou and
  Marchant}{2016}]{Bouyssou}
---\hspace{-.1pt}---\hspace{-.1pt}--- (2016): \enquote{Ranking authors using
  fractional counting of citations: An axiomatic approach,} \emph{Journal of
  Informetrics}, 10, 183--199.

\bibitem[\protect\citeauthoryear{Brin and Page}{Brin and
  Page}{1998}]{brin1998anatomy}
\textsc{Brin, S. and L.~Page} (1998): \enquote{The anatomy of a large-scale
  hypertextual web search engine,} \emph{Computer networks and ISDN systems},
  30, 107--117.

\bibitem[\protect\citeauthoryear{Chambers and Miller}{Chambers and
  Miller}{2014}]{Chambers}
\textsc{Chambers, C.~P. and A.~D. Miller} (2014): \enquote{Scholarly
  influence,} \emph{Journal of Economic Theory}, 151, 571--583.

\bibitem[\protect\citeauthoryear{Checchi, Ciolfi, De~Fraja, Mazzotta, and
  Verzillo}{Checchi et~al.}{2021}]{checchi2021have}
\textsc{Checchi, D., A.~Ciolfi, G.~De~Fraja, I.~Mazzotta, and S.~Verzillo}
  (2021): \enquote{Have you read this? An empirical comparison of the British
  REF peer review and the Italian VQR bibliometric algorithm,}
  \emph{Economica}, 88, 1107--1129.

\bibitem[\protect\citeauthoryear{De~La~Vega and Volij}{De~La~Vega and
  Volij}{2018}]{de2018ranking}
\textsc{De~La~Vega, C.~L. and O.~Volij} (2018): \enquote{Ranking scholars: A
  measure representation,} \emph{Journal of Informetrics}, 12, 510--517.

\bibitem[\protect\citeauthoryear{Ellison}{Ellison}{2013}]{ellison2013does}
\textsc{Ellison, G.} (2013): \enquote{How does the market use citation data?
  The Hirsch index in economics,} \emph{American Economic Journal: Applied
  Economics}, 5, 63--90.

\bibitem[\protect\citeauthoryear{Flores-Szwagrzak and
  Treibich}{Flores-Szwagrzak and Treibich}{2020}]{flores2020teamwork}
\textsc{Flores-Szwagrzak, K. and R.~Treibich} (2020): \enquote{Teamwork and
  individual productivity,} \emph{Management Science}, 66, 2523--2544.

\bibitem[\protect\citeauthoryear{Garfield}{Garfield}{1972}]{Garfield}
\textsc{Garfield, E.} (1972): \enquote{Citation Analysis as a Tool in Journal
  Evaluation,} \emph{Science}, 178, 471--479.

\bibitem[\protect\citeauthoryear{Hamermesh}{Hamermesh}{2013}]{Hamermesh}
\textsc{Hamermesh, D.~S.} (2013): \enquote{Six decades of top economics
  publishing: Who and how?} \emph{Journal of Economic Literature}, 51,
  162--172.

\bibitem[\protect\citeauthoryear{Hamermesh}{Hamermesh}{2018}]{hamermesh2018citations}
---\hspace{-.1pt}---\hspace{-.1pt}--- (2018): \enquote{Citations in economics:
  Measurement, uses, and impacts,} \emph{Journal of Economic Literature}, 56,
  115--156.

\bibitem[\protect\citeauthoryear{Hicks, Wouters, Waltman, De~Rijcke, and
  Rafols}{Hicks et~al.}{2015}]{hicks2015bibliometrics}
\textsc{Hicks, D., P.~Wouters, L.~Waltman, S.~De~Rijcke, and I.~Rafols} (2015):
  \enquote{Bibliometrics: the Leiden Manifesto for research metrics,}
  \emph{Nature}, 520, 429--431.

\bibitem[\protect\citeauthoryear{Hirsch}{Hirsch}{2005}]{Hirsch}
\textsc{Hirsch, J.~E.} (2005): \enquote{An index to quantify an individual's
  scientific research output,} \emph{Proceedings of the National academy of
  Sciences}, 16569--16572.

\bibitem[\protect\citeauthoryear{Ioannidis, Boyack, and Wouters}{Ioannidis
  et~al.}{2016}]{Ioannidis}
\textsc{Ioannidis, J.~P., K.~Boyack, and P.~F. Wouters} (2016):
  \enquote{Citation Metrics: A primer on how (not) to normalize,} \emph{PLoS
  biology}, 14, e1002542.

\bibitem[\protect\citeauthoryear{Kaur, Radicchi, and Menczer}{Kaur
  et~al.}{2013}]{Kaur}
\textsc{Kaur, J., F.~Radicchi, and F.~Menczer} (2013): \enquote{Universality of
  scholarly impact metrics,} \emph{Journal of Informetrics}, 7, 924--932.

\bibitem[\protect\citeauthoryear{Klosik and Bornholdt}{Klosik and
  Bornholdt}{2014}]{Klosik}
\textsc{Klosik, D.~F. and S.~Bornholdt} (2014): \enquote{The citation wake of
  publications detects Nobel laureates' papers,} \emph{PloS one}, 9, e113184.

\bibitem[\protect\citeauthoryear{K{\'o}czy and Strobel}{K{\'o}czy and
  Strobel}{2009}]{Koczy}
\textsc{K{\'o}czy, L. and M.~Strobel} (2009): \enquote{The invariant method can
  be manipulated,} \emph{Scientometrics}, 81, 291--293.

\bibitem[\protect\citeauthoryear{Levene, Fenner, and Bar-Ilan}{Levene
  et~al.}{2019}]{levene2019characterisation}
\textsc{Levene, M., T.~Fenner, and J.~Bar-Ilan} (2019):
  \enquote{Characterisation of the $\chi$-index and the rec-index,}
  \emph{Scientometrics}, 120, 885--896.

\bibitem[\protect\citeauthoryear{Mac{\'e}}{Mac{\'e}}{2017}]{mace2017measuring}
\textsc{Mac{\'e}, A.} (2017): \enquote{Measuring influence in science: Standing
  on the shoulders of which giants?} \emph{arXiv preprint arXiv:1711.02695}.

\bibitem[\protect\citeauthoryear{Marchant}{Marchant}{2009}]{Marchant}
\textsc{Marchant, T.} (2009): \enquote{An axiomatic characterization of the
  ranking based on the h-index and some other bibliometric rankings of
  authors,} \emph{Scientometrics}, 80, 325--342.

\bibitem[\protect\citeauthoryear{Palacios-Huerta and Volij}{Palacios-Huerta and
  Volij}{2004}]{Palacios2004}
\textsc{Palacios-Huerta, I. and O.~Volij} (2004): \enquote{The measurement of
  intellectual influence,} \emph{Econometrica}, 72, 963--977.

\bibitem[\protect\citeauthoryear{Perry and Reny}{Perry and Reny}{2016}]{Perry}
\textsc{Perry, M. and P.~J. Reny} (2016): \enquote{How to count citations if
  you must,} \emph{The American Economic Review}, 106, 2722--2741.

\bibitem[\protect\citeauthoryear{Pinski and Narin}{Pinski and
  Narin}{1976}]{Pinski}
\textsc{Pinski, G. and F.~Narin} (1976): \enquote{Citation influence for
  journal aggregates of scientific publications: Theory, with application to
  the literature of physics,} \emph{Information processing \& management}, 12,
  297--312.

\bibitem[\protect\citeauthoryear{Radicchi, Fortunato, and Castellano}{Radicchi
  et~al.}{2008}]{Radicchi2008}
\textsc{Radicchi, F., S.~Fortunato, and C.~Castellano} (2008):
  \enquote{Universality of citation distributions: Toward an objective measure
  of scientific impact,} \emph{Proceedings of the National Academy of
  Sciences}, 105, 17268--17272.

\bibitem[\protect\citeauthoryear{Radicchi, Fortunato, Markines, and
  Vespignani}{Radicchi et~al.}{2009}]{Radicchi2009}
\textsc{Radicchi, F., S.~Fortunato, B.~Markines, and A.~Vespignani} (2009):
  \enquote{Diffusion of scientific credits and the ranking of scientists,}
  \emph{Physical Review E}, 80, 056103.

\bibitem[\protect\citeauthoryear{Scotchmer}{Scotchmer}{1991}]{Scotchmer}
\textsc{Scotchmer, S.} (1991): \enquote{Standing on the shoulders of giants:
  cumulative research and the patent law,} \emph{The journal of economic
  perspectives}, 5, 29--41.

\bibitem[\protect\citeauthoryear{Seeber, Cattaneo, Meoli, and
  Malighetti}{Seeber et~al.}{2019}]{seeber2019self}
\textsc{Seeber, M., M.~Cattaneo, M.~Meoli, and P.~Malighetti} (2019):
  \enquote{Self-citations as strategic response to the use of metrics for
  career decisions,} \emph{Research Policy}, 48, 478--491.

\bibitem[\protect\citeauthoryear{Shapley}{Shapley}{1953}]{Shapley}
\textsc{Shapley, L.~S.} (1953): \enquote{A value for n-person games,}
  \emph{Contributions to the Theory of Games}, 2, 307--317.

\bibitem[\protect\citeauthoryear{Waltman and van Eck}{Waltman and van
  Eck}{2010}]{Waltman2010}
\textsc{Waltman, L. and N.~J. van Eck} (2010): \enquote{The relation between
  Eigenfactor, audience factor, and influence weight,} \emph{Journal of the
  Association for Information Science and Technology}, 61, 1476--1486.

\bibitem[\protect\citeauthoryear{Waltman and van Eck}{Waltman and van
  Eck}{2019}]{waltman2019field}
---\hspace{-.1pt}---\hspace{-.1pt}--- (2019): \enquote{Field normalization of
  scientometric indicators,} \emph{Springer handbook of science and technology
  indicators}, 281--300.

\bibitem[\protect\citeauthoryear{Waltman, van Eck, van Leeuwen, and
  Visser}{Waltman et~al.}{2013}]{Waltman2013}
\textsc{Waltman, L., N.~J. van Eck, T.~N. van Leeuwen, and M.~S. Visser}
  (2013): \enquote{Some modifications to the SNIP journal impact indicator,}
  \emph{Journal of informetrics}, 7, 272--285.

\bibitem[\protect\citeauthoryear{Wang, Shen, and Cheng}{Wang
  et~al.}{2016}]{Wang}
\textsc{Wang, H., H.-W. Shen, and X.-Q. Cheng} (2016): \enquote{Scientific
  credit diffusion: Researcher level or paper level?} \emph{Scientometrics},
  109, 827--837.

\bibitem[\protect\citeauthoryear{West, Jensen, Dandrea, Gordon, and
  Bergstrom}{West et~al.}{2013}]{West}
\textsc{West, J.~D., M.~C. Jensen, R.~J. Dandrea, G.~J. Gordon, and C.~T.
  Bergstrom} (2013): \enquote{Author-level Eigenfactor metrics: Evaluating the
  influence of authors, institutions, and countries within the social science
  research network community,} \emph{Journal of the Association for Information
  Science and Technology}, 64, 787--801.

\bibitem[\protect\citeauthoryear{Woeginger}{Woeginger}{2008{\natexlab{a}}}]{woeginger2008axiomatic}
\textsc{Woeginger, G.~J.} (2008{\natexlab{a}}): \enquote{An axiomatic analysis
  of Egghe's g-index,} \emph{Journal of Informetrics}, 2, 364--368.

\bibitem[\protect\citeauthoryear{Woeginger}{Woeginger}{2008{\natexlab{b}}}]{Woeginger}
---\hspace{-.1pt}---\hspace{-.1pt}--- (2008{\natexlab{b}}): \enquote{An
  axiomatic characterization of the Hirsch-index,} \emph{Mathematical Social
  Sciences}, 56, 224--232.

\bibitem[\protect\citeauthoryear{Wuchty, Jones, and Uzzi}{Wuchty
  et~al.}{2007}]{Wuchty}
\textsc{Wuchty, S., B.~F. Jones, and B.~Uzzi} (2007): \enquote{The increasing
  dominance of teams in production of knowledge,} \emph{Science}, 316,
  1036--1039.

\bibitem[\protect\citeauthoryear{Zitt, Ramanana-Rahary, and Bassecoulard}{Zitt
  et~al.}{2005}]{Zitt2005}
\textsc{Zitt, M., S.~Ramanana-Rahary, and E.~Bassecoulard} (2005):
  \enquote{Relativity of citation performance and excellence measures: From
  cross-field to cross-scale effects of field-normalisation,}
  \emph{Scientometrics}, 63, 373--401.

\end{thebibliography}

\newpage

\section{Appendix}

\subsection{Proof of \theoremref{tMain}}
We show a first lemma, establishing that Author Anonymity extends to the case where the number of authors change when we reallocate authors to papers.

\begin{sloppypar}
\begin{lemma}\label{lEAA}
If a neutral and separable influence measure satisfies Author Anonymity, then it satisfies Extended Author Anonymity: for any databases $d=(\mathcal{A},\mathcal{P}_{\mathcal{A}},n)$ and  $d'=(\mathcal{A}',\mathcal{P}'_{\mathcal{A}'},n)$ in $\mathbb{D}$ such that $a\in \mathcal{A}\cap \mathcal{A}'$ and $\mathcal{P}=\mathcal{P}'$, we have $f_a(d')=f_a(d)$.
\end{lemma}
\end{sloppypar}

\begin{proof}
For $\#\mathcal{A}=\#\mathcal{A}'$, the result is a direct consequence of neutrality and Author Anonymity. To prove the lemma, it is sufficient to prove that it holds for $\#\mathcal{A}'=1+\#\mathcal{A}$ (the lemma then follows by induction). By neutrality, we may even assume that $\mathcal{A}\subset \mathcal{A}'$, so that we have $\mathcal{A}'=\mathcal{A}\cup \{u\}$, with $u\notin \mathcal{A}$. The strategy of the proof consists in adding an auxiliary database $d_{aux}$ to $d$, and in applying Author Anonymity to $d\oplus d_{aux}$.

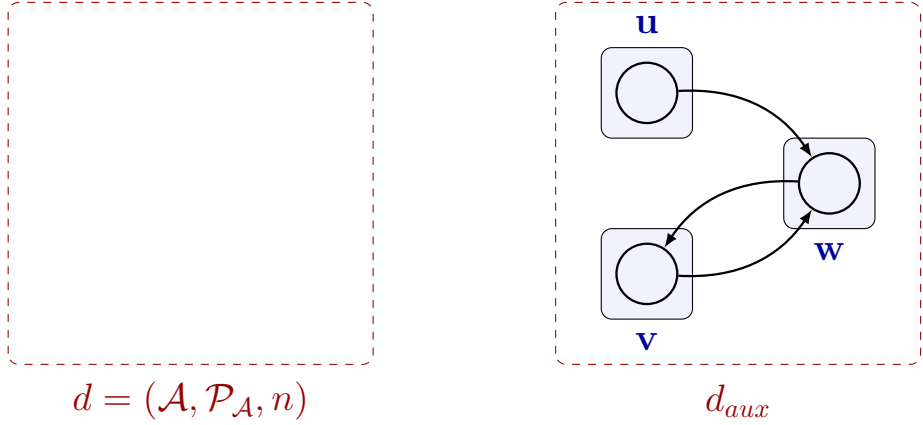
\begin{figure}[!h]
    \centering
    \begin{subfigure}[t]{\textwidth}
        \centering
        \begin{center}
\begin{tikzpicture}[scale=.5]
\draw[rounded corners,dashed, color=gray] (-10,-2.5) rectangle (14,8);
\draw[rounded corners,dashed, color=MyRed] (-8,-1) rectangle (0,7);
\draw[color=MyRed] (-4,-1.8) node{\large $d=(\mathcal{A},\mathcal{P}_{\mathcal{A}},n)$};

\draw[rounded corners,dashed, color=MyRed] (4,-1) rectangle (12,7);
\draw[color=MyRed] (8,-1.8) node{\large $d_{aux}$};

\draw[rounded corners,fill=blue!05] (5,0) rectangle (7,2);
\draw[color=MyBlue] (6,-0.5) node{\large $\bf{v}$};
\node[draw,circle, line width=.3mm, minimum size=.7cm] (Pv) at (6,1){};

\draw[rounded corners,fill=blue!05] (5,4) rectangle (7,6);
\draw[color=MyBlue] (6,6.5) node{\large $\bf{u}$};
\node[draw,circle, line width=.3mm, minimum size=.7cm] (Pu) at (6,5){};

\draw[rounded corners,fill=blue!05] (9,2) rectangle (11,4);
\draw[color=MyBlue] (10,1.5) node{\large $\bf{w}$};
\node[draw,circle, line width=.3mm, minimum size=.7cm] (Pw) at (10,3){};

\draw[->,>=latex, line width=.3mm] (Pu) to[bend left] (Pw);
\draw[->,>=latex, line width=.3mm] (Pv) to[bend right] (Pw);
\draw[->,>=latex, line width=.3mm] (Pw) to[bend right] (Pv);
\end{tikzpicture}
\end{center}
        \caption{Database $d\oplus d_{aux}$} \label{fLemma1.1}
    \end{subfigure}

    \vspace{.5cm}
    \begin{subfigure}[t]{\textwidth}
    \centering
    \begin{center}
\begin{tikzpicture}[scale=.5]
\draw[rounded corners,dashed, color=gray] (-10,-2.5) rectangle (14,8);
\draw[rounded corners,dashed, color=MyRed] (-8,-1) rectangle (0,7);
\draw[color=MyRed] (-4,-1.8) node{\large $d'=(\mathcal{A}\cup\{u\},\mathcal{P}'_{\mathcal{A}'},n)$};

\draw[rounded corners,dashed, color=MyRed] (4,-1) rectangle (12,7);
\draw[color=MyRed] (8,-1.8) node{\large $d'_{aux}$};

\draw[rounded corners,fill=blue!05] (5,0) rectangle (7,6);
\draw[color=MyBlue] (6,-0.5) node{\large $\bf{v}$};
\node[draw,circle, line width=.3mm, minimum size=.7cm] (Pv) at (6,1){};
\node[draw,circle, line width=.3mm, minimum size=.7cm] (Pu) at (6,5){};

\draw[rounded corners,fill=blue!05] (9,2) rectangle (11,4);
\draw[color=MyBlue] (10,1.5) node{\large $\bf{w}$};
\node[draw,circle, line width=.3mm, minimum size=.7cm] (Pw) at (10,3){};

\draw[->,>=latex, line width=.3mm] (Pu) to[bend left] (Pw);
\draw[->,>=latex, line width=.3mm] (Pv) to[bend right] (Pw);
\draw[->,>=latex, line width=.3mm] (Pw) to[bend right] (Pv);
\end{tikzpicture}
\end{center}
        \caption{Database $d'\oplus d'_{aux}$} \label{fLemma1.2}
    \end{subfigure}
    \caption{Databases in the proof of \lemmaref{lEAA}} \label{fLemma1}
\end{figure}

\FloatBarrier

The auxiliary database $d_{aux}$ is represented in \figref{fLemma1.1}. It contains three authors $u,v,w$, and each author is represented by a rounded square. Each author has a single paper, represented by a circle. Each citation is represented by an arrow (we can read that $u$'s paper cites $w$'s paper and that $v$'s and $w$'s papers cite each other). By Author Anonymity, we can reallocate authors to papers, while keeping the set of authors, the set of papers and the network fixed. We obtain the database represented in \figref{fLemma1.2}.

In the database $d'_{aux}$, there are only two authors $v$ and $w$, and $v$ has written two papers. As we have $d_{aux},d'_{aux}\in \mathbb{D}$, we can write, using Separability and Author Anonymity:
$$f_a(d)=f_a(d\oplus d_{aux})= f_a(d'\oplus d'_{aux})=f_a(d').$$
This concludes the proof of the lemma.
\end{proof}

We now prove \theoremref{tMain}.

\begin{proof}

It is clear that any citation-counting index is neutral and satisfies the five axioms. Conversely, let $f$ be a neutral influence measure satisfying the five axioms. Let $d=(\mathcal{A}, \mathcal{P}_{\mathcal{A}},n)$ be a database in $\mathbb{D}$, and let $a\in \mathcal{A}$. We will show that the influence of $a$ in $d$, $f_a(d)$, equals the influence of $a$ in any database $d'=(\mathcal{A}', \mathcal{P}'_{\mathcal{A}'},n')$ such that:
\begin{itemize}
	\item $A'=\{a,b\}$
	\item $\mathcal{P}_a'=\mathcal{P}_a$ and  $P'_b=\sum_{p\in \mathcal{P}_a}C_p$ ($b$ has written a number of papers equal to the total number of citations of $a$ in $d$)
	\item $\forall p\in \mathcal{P}_a$, $C'_p=C_p$ (each of $a$'s papers keeps the same number of citations), while $\forall q\in \mathcal{P}'_b$, $R'_q=1$ (each of $b$'s papers has a unique reference)
	\item $\exists (p,q) \in \mathcal{P}'_b\times \mathcal{P}_a',\quad n'(p,q)=1$ (there is at least a citation from $a$ to $b$).
\end{itemize}
Note that the last line is required to have $d'\in \mathbb{D}$. Since $f$ satisfies Reference Independence, the identities of the papers from $a$ citing $b$, as well as these of the papers cited by $a$, should not matter for the score of $a$. As $f$ is neutral, the names of the papers in $\mathcal{P}'_b$ should  not matter for the score of $a$ either. Therefore, all  databases of the form of $d'$ give the same score to $a$.

Let $b\in \mathcal{A}$ be an author citing $a$ in the database $d$: there exists a paper $p_b\in \mathcal{P}_b$ with $\mathcal{R}_{p_b}\cap \mathcal{P}_a\neq \emptyset$. Let us first assume that $p_b$ is the sole paper from $b$, that is, $\mathcal{P}_b=\{p_b\}$
. The essential steps of the proof can indeed be presented under this assumption, and we will argue later that they can be generalized. We also illustrate the main steps of the proof on \figref{fTheorem1}. 

\begin{figure}[!h]
    \centering
    \begin{subfigure}[t]{.45\textwidth}
        \centering
		\begin{center}
\begin{tikzpicture}[scale=.5]
\draw[rounded corners,fill=blue!05] (1,0) rectangle (3,8);
\draw[color=MyBlue] (2,-0.5) node{\large $\bf{a}$};
\node[draw,circle, line width=.3mm, minimum size=.7cm] (Pa1) at (2,2){};
\node[draw,circle, line width=.3mm, minimum size=.7cm]  (Pa2) at (2,6){};

\draw[rounded corners,fill=blue!05] (5,0) rectangle (7,2);
\draw[color=MyBlue] (6,-0.5) node{\large $\bf{b}$};
\node[draw,circle, line width=.3mm, minimum size=.7cm] (Pb) at (6,1){};
\node (Pbr) at (4,1){};
\node (Pbc) at (7.5,-1){};
\draw[->,>=latex, line width=.3mm] (Pb) to[bend right] (Pa1);
\draw[->,>=latex, line width=.3mm] (Pb) to[bend right] (Pa2);
\draw[->,>=latex, line width=.3mm] (Pb) to[bend right] (Pbr);
\draw[->,>=latex, line width=.3mm] (Pbc) to[bend right] (Pb);

\draw[rounded corners,fill=blue!05] (5,6) rectangle (11,8);
\draw[color=MyBlue] (8,8.5) node{\large $\bf{u}$};
\node[draw,circle, line width=.3mm, minimum size=.7cm] (Pu1) at (6,7){};
\node[draw,circle, line width=.3mm, minimum size=.7cm] (Pu2) at (10,7){};

\draw[rounded corners,fill=blue!05] (9,0) rectangle (11,2);
\draw[color=MyBlue] (10,-0.5) node{\large $\bf{v}$};
\node[draw,circle, line width=.3mm, minimum size=.7cm] (Pv) at (10,1){};
\draw[->,>=latex, line width=.3mm] (Pu2) to[bend right] (Pv);
\draw[->,>=latex, line width=.3mm] (Pv) to[bend right] (Pu2);
\end{tikzpicture}
\end{center}
        \caption{Step 1} \label{fStep1}
    	\end{subfigure}
    	\hfill
        \begin{subfigure}[t]{.45\textwidth}
        \centering
\begin{center}
\begin{tikzpicture}[scale=.5]
\draw[rounded corners,fill=blue!05] (1,0) rectangle (3,8);
\draw[color=MyBlue] (2,-0.5) node{\large $\bf{a}$};
\node[draw,circle, line width=.3mm, minimum size=.7cm] (Pa1) at (2,2){};
\node[draw,circle, line width=.3mm, minimum size=.7cm]  (Pa2) at (2,6){};

\draw[rounded corners,fill=blue!05] (5,0) rectangle (7,8);
\draw[color=MyBlue] (6,-0.5) node{\large $\bf{b}$};
\node[draw,circle, line width=.3mm, minimum size=.7cm] (Pb) at (6,1){};
\node (Pbr) at (4,1){};
\node (Pbc) at (7.5,-1){};
\draw[->,>=latex, line width=.3mm] (Pb) to[bend right] (Pa1);
\draw[->,>=latex, line width=.3mm] (Pb) to[bend right] (Pa2);
\draw[->,>=latex, line width=.3mm] (Pb) to[bend right] (Pbr);
\draw[->,>=latex, line width=.3mm] (Pbc) to[bend right] (Pb);

\draw[rounded corners,fill=blue!05] (9,6) rectangle (11,8);
\draw[color=MyBlue] (10,8.5) node{\large $\bf{u}$};
\node[draw,circle, line width=.3mm, minimum size=.7cm] (Pu1) at (6,7){};
\node[draw,circle, line width=.3mm, minimum size=.7cm] (Pu2) at (10,7){};

\draw[rounded corners,fill=blue!05] (9,0) rectangle (11,2);
\draw[color=MyBlue] (10,-0.5) node{\large $\bf{v}$};
\node[draw,circle, line width=.3mm, minimum size=.7cm] (Pv) at (10,1){};
\draw[->,>=latex, line width=.3mm] (Pu2) to[bend right] (Pv);
\draw[->,>=latex, line width=.3mm] (Pv) to[bend right] (Pu2);
\end{tikzpicture}
\end{center}
        \caption{Step 2} \label{fStep2}
   		\end{subfigure}
   		
   		\vspace{.5cm}
       	
        \begin{subfigure}[t]{.45\textwidth}
        \centering
\begin{center}
\begin{tikzpicture}[scale=.5]
\draw[rounded corners,fill=blue!05] (1,0) rectangle (3,8);
\draw[color=MyBlue] (2,-0.5) node{\large $\bf{a}$};
\node[draw,circle, line width=.3mm, minimum size=.7cm] (Pa1) at (2,2){};
\node[draw,circle, line width=.3mm, minimum size=.7cm]  (Pa2) at (2,6){};

\draw[rounded corners,fill=blue!05] (5,0) rectangle (7,8);
\draw[color=MyBlue] (6,-0.5) node{\large $\bf{b}$};
\node[draw,circle, line width=.3mm, minimum size=.7cm] (Pb) at (6,1){};
\node (Pbr) at (4,1){};
\node (Pbc) at (7.5,-1){};

\draw[rounded corners,fill=blue!05] (9,6) rectangle (11,8);
\draw[color=MyBlue] (10,8.5) node{\large $\bf{u}$};
\node[draw,circle, line width=.3mm, minimum size=.7cm] (Pu1) at (6,7){};
\node[draw,circle, line width=.3mm, minimum size=.7cm] (Pu2) at (10,7){};
\draw[->,>=latex, line width=.3mm] (Pu1) to[bend right] (Pa1);
\draw[->,>=latex, line width=.3mm] (Pu1) to[bend right] (Pa2);
\draw[->,>=latex, line width=.3mm] (Pu1) to[bend right] (Pbr);
\draw[->,>=latex, line width=.3mm] (Pbc) to[bend right] (Pb);

\draw[rounded corners,fill=blue!05] (9,0) rectangle (11,2);
\draw[color=MyBlue] (10,-0.5) node{\large $\bf{v}$};
\node[draw,circle, line width=.3mm, minimum size=.7cm] (Pv) at (10,1){};
\draw[->,>=latex, line width=.3mm] (Pu2) to[bend right] (Pv);
\draw[->,>=latex, line width=.3mm] (Pv) to[bend right] (Pu2);
\end{tikzpicture}
\end{center}
        \caption{Step 3} \label{fStep3}
   		\end{subfigure}
	\hfill
    \begin{subfigure}[t]{.45\textwidth}
    \centering
\begin{center}
\begin{tikzpicture}[scale=.5]
\draw[rounded corners,fill=blue!05] (1,0) rectangle (3,8);
\draw[color=MyBlue] (2,-0.5) node{\large $\bf{a}$};
\node[draw,circle, line width=.3mm, minimum size=.7cm] (Pa1) at (2,2){};
\node[draw,circle, line width=.3mm, minimum size=.7cm]  (Pa2) at (2,6){};

\draw[rounded corners,fill=blue!05] (5,0) rectangle (7,8);
\draw[color=MyBlue] (6,-0.5) node{\large $\bf{b}$};
\node[draw,circle, line width=.3mm, minimum size=.7cm] (Pb) at (6,1){};
\node[draw,circle, line width=.3mm, minimum size=.7cm] (Pb2) at (6,2.5){};
\node[draw,circle, line width=.3mm, minimum size=.7cm] (Pb3) at (6,5.5){};
\node (Pbr) at (4,1){};
\node (Pbc) at (7.5,-1){};

\draw[rounded corners,fill=blue!05] (9,6) rectangle (11,8);
\draw[color=MyBlue] (10,8.5) node{\large $\bf{u}$};
\node[draw,circle, line width=.3mm, minimum size=.7cm] (Pu1) at (6,7){};
\node[draw,circle, line width=.3mm, minimum size=.7cm] (Pu2) at (10,7){};

\draw[->,>=latex, line width=.3mm] (Pu1) to[bend right] (Pa2);
\draw[->,>=latex, line width=.3mm] (Pb3) to[bend right] (Pa1);
\draw[->,>=latex, line width=.3mm] (Pb2) to[bend right] (Pbr);
\draw[->,>=latex, line width=.3mm] (Pbc) to[bend right] (Pb);

\draw[rounded corners,fill=blue!05] (9,0) rectangle (11,2);
\draw[color=MyBlue] (10,-0.5) node{\large $\bf{v}$};
\node[draw,circle, line width=.3mm, minimum size=.7cm] (Pv) at (10,1){};
\draw[->,>=latex, line width=.3mm] (Pu2) to[bend right] (Pv);
\draw[->,>=latex, line width=.3mm] (Pv) to[bend right] (Pu2);
\end{tikzpicture}
\end{center}
        \caption{Step 4} \label{fStep4}
    \end{subfigure}
    
	\vspace{.5cm}    
    
        \begin{subfigure}[t]{.45\textwidth}
    \centering
\begin{center}
\begin{tikzpicture}[scale=.5]

\draw[rounded corners,fill=blue!05] (1,0) rectangle (3,8);
\draw[color=MyBlue] (2,-0.5) node{\large $\bf{a}$};
\node[draw,circle, line width=.3mm, minimum size=.7cm] (Pa1) at (2,2){};
\node[draw,circle, line width=.3mm, minimum size=.7cm]  (Pa2) at (2,6){};

\draw[rounded corners,fill=blue!05] (5,4.5) rectangle (7,8);
\draw[color=MyBlue] (6,8.5) node{\large $\bf{b}$};

\draw[rounded corners,fill=blue!05] (5,0) rectangle (11,3.5);
\draw[color=MyBlue] (8,-0.5) node{\large $\bf{v}$}; 

\node[draw,circle, line width=.3mm, minimum size=.7cm] (Pb) at (6,1){};
\node[draw,circle, line width=.3mm, minimum size=.7cm] (Pb2) at (6,2.5){};
\node[draw,circle, line width=.3mm, minimum size=.7cm] (Pb3) at (6,5.5){};
\node (Pbr) at (4,1){};
\node (Pbc) at (7.5,-1){};

\draw[rounded corners,fill=blue!05] (9,6) rectangle (11,8);
\draw[color=MyBlue] (10,8.5) node{\large $\bf{u}$};
\node[draw,circle, line width=.3mm, minimum size=.7cm] (Pu1) at (6,7){};
\node[draw,circle, line width=.3mm, minimum size=.7cm] (Pu2) at (10,7){};

\draw[->,>=latex, line width=.3mm] (Pu1) to[bend right] (Pa2);
\draw[->,>=latex, line width=.3mm] (Pb3) to[bend right] (Pa1);
\draw[->,>=latex, line width=.3mm] (Pb2) to[bend right] (Pbr);
\draw[->,>=latex, line width=.3mm] (Pbc) to[bend right] (Pb);

\node[draw,circle, line width=.3mm, minimum size=.7cm] (Pv) at (10,1){};
\draw[->,>=latex, line width=.3mm] (Pu2) to[bend right] (Pv);
\draw[->,>=latex, line width=.3mm] (Pv) to[bend right] (Pu2);
\end{tikzpicture}
\end{center}
        \caption{Step 5} \label{fStep5}
    \end{subfigure}
    \caption{Steps for the proof of \theoremref{tMain}} \label{fTheorem1}
\end{figure}

\FloatBarrier

\textbf{Step 1.} We add to $d$ an auxiliary database with two authors, $u$ and $v$, as represented on \figref{fStep1}. Formally, $\mathcal{A}^{aux}=\{u,v\}$ with $\mathcal{P}^{aux}_u=\{p_u,p_u'\}$ and  $\mathcal{P}^{aux}_v=\{p_v\}$ and $n^{aux}(p_u',p_v)=n^{aux}(p_v,p_u')=1$. We know by Separability that this does not affect the score of $a$.

\textbf{Step 2.} We modify the database to transfer the paper with no citations and no references from $u$ to $b$, as represented on \figref{fStep2}. We now have $\mathcal{P}^2_b=\{p_b,p_u\}$ and $\mathcal{P}^2_u=\{p_u'\}$, where use the notation $\mathcal{P}^m_x$ (resp. $\mathcal{A}^m$ and $n^m$) for the set of papers from author $x$ in the joint database at step $m$. By Author Anonymity, this does not affect the score of $a$.

\textbf{Step 3.} We modify the network so that all the citations from $b$ now originate from the paper $p_u$ (rather than from $p_b$), as represented on \figref{fStep3}. Formally, $n^3(p,p_b)=0$ and $n^3(p,p_u)=n^2(p,p_b)$ for any $p$; and $n^3(p,q)=n^2(p,q)$ whenever $q\neq p_b,p_u$. By Citation Anonymity, this does not affect the score of $a$ (since both $p_b$ and $p_u$ belong to $\mathcal{P}_b^2=\mathcal{P}_b^3$).

\textbf{Step 4.} We split the (uncited) paper $p_u$ from $b$ into $k+1$ papers, where $k$ is the number of references in $p_u$ to papers written by $a$. Each of the first $k$ papers cites one (and only one) paper from $a$, while the remaining paper keeps all references to papers not written by $a$,  as represented on \figref{fStep4}. Formally, $k=\# (\mathcal{R}^3_{p_u}\cap \mathcal{P}^3_a)$ and  $\mathcal{P}_b^4=\{p^{(1)}, \ldots, p^{(k)}, p^{(k+1)}, p_b\}$ with $\cup_{i=1}^k\mathcal{R}^4_{p^{i}} = \mathcal{R}^3_{p_u}\cap \mathcal{P}^3_a$ and $\forall i=1\ldots k$, $R^4_{p^{i}}=1$ and $\mathcal{R}^4_{p^{k+1}} = \mathcal{R}^3_{p_u}\backslash \mathcal{P}^3_a$. Since the split paper $p_u$ had no citation, we can apply Splitting, and we obtain that the score of $a$ remains unaffected.

\textbf{Step 5.} We isolate the papers from $b$ citing $a$, by having $v$ absorbing the two other papers written by $b$, namely $p^{(k+1)}$ and $p_b$,  as represented on \figref{fStep5}. Formally, $\mathcal{P}_b^5=\{p^{(1)}, \ldots, p^{(k)}\}$ and $\mathcal{P}_v^5=\{p^{(k+1)}, p_b,p_v\}$. By Author Anonymity, this does not affect the score of $a$.

We have assumed that there was initially a single paper from $b$ citing $a$. Observe now that the construction can be iterated if there are multiple papers from $b$ citing $a$. In the resulting database, there are as many papers from $b$ as the initial number of citations from $b$ to $a$ (each has one references), and there are no other citations either from or to papers written by $b$. 

If there are initially multiple authors citing $a$, we can iterate the construction, and then merge all authors citing at least one of $a$'s papers to a single author (we will call it $b$). By application of Extended Author Anonymity (which holds by \lemmaref{lEAA}), this does not affect the score of $a$.

Finally, we can cut all references from $a$ to authors different from $b$ and add one citation from (any paper of) $a$ to (any paper of) $b$. By Reference Independence, this does not affect the score of $a$.

To conclude, we have obtained a database whose papers can be partitioned in three parts: the papers written by $a$, the papers citing $a$, all being written by the same author $b$, and all the other papers. In this database, there is no citation from either $a$ or $b$ to the remaining authors, nor from the remaining authors to either $a$ or $b$. Moreover both the sub-database containing papers from $a$ and $b$, and the remaining sub-database, belong to the set $\mathbb{D}$, they thus form two disjoint fields. By application of Separability, we obtain that the score of $a$ is the same as in the former sub-database, noted $d'$, which has all the properties listed at the beginning of the proof. 

To conclude, we obtained that if two databases are such that $m_a[d]=m_a[d'']$, there exists a database $d'$ (obtained from our construction) such that $f_a(d)=f_a(d')=f_a(d'')$. Thus, $f$ is a citation-counting index.\medskip

\textbf{Independence of the axioms}. For each axiom, we propose a measure $f$ that satisfies all axioms but this one. 
\begin{itemize}
	\item Separability. The score of an author $a$ equals the ratio between the number of citations received by $a$ and the total number of references of authors other than $a$:
	$$f_a(d)=\dfrac{\sum_{p\in \mathcal{P}_a} C_p}{\sum_{p\in \mathcal{P}\backslash \mathcal{P}_a} R_p}.$$
	\item Reference Independence. The score of each author equals her number of references: $f_a(d)=\sum_{p\in \mathcal{P}_a}R_p$.
	\item Splitting. The score of each author is her fractional citation count (see \footnoteref{fnFCC}).
	\item Citation Anonymity. The score of an author $a$ is the total number of citations that papers citing a paper from $a$ receive, and that are not issued by $a$:
	$$f_a(d)=\sum_{q\in \cup_{p\in\mathcal{P}_a}\mathcal{C}_p}\# \left(\mathcal{C}_q\cap (\mathcal{P}\backslash \mathcal{P}_a)\right).$$
    \item Author Anonymity.  The score of $a$ is the number of authors citing $a$:
    $$f_a(d)=\#\{b\in \mathcal{A},\ (\cup_{q\in \mathcal{P}_b}\mathcal{R}_q )\cap \mathcal{P}_a \neq \emptyset  \}.$$ 
\end{itemize}
\end{proof}

\subsection{Proof of \propref{pImpossibility}}

\begin{proof}
Let $f$ be a non-degenerate and separable influence measure satisfying Null Author, Author Anonymity and Field Comparability. As $f$ is non-degenerate, there is $d^0=(\mathcal{A}^0,\mathcal{P}^0_{\mathcal{A}^0},n^0)$ and $a^0\in \mathcal{A}^0$ such that $f_{a^0}(d^0)>0$. For any $d=(\mathcal{A},\mathcal{P}_{\mathcal{A}},n)$ such that $\mathcal{A}^0\cap \mathcal{A}=\emptyset$ and $\mathcal{P}^0\cap \mathcal{P}=\emptyset$, we obtain by Separability and Field Comparability that:
$$\dfrac{1}{A^0}\sum_{a\in \mathcal{A}^0}f_a(d^0) = \dfrac{1}{A^0}\sum_{a\in \mathcal{A}^0}f_a(d^0\oplus d) =  \dfrac{1}{A}\sum_{a\in \mathcal{A}}f_a(d^0\oplus d) = \dfrac{1}{A}\sum_{a\in \mathcal{A}}f_a( d).$$
It follows that  $\sum_{a\in \mathcal{A}}f_a( d)>0$.

Consider now the database $d\in \mathbb{D}$ represented in \figref{fProp1.1}, where we assume that the names of authors and papers from $d$ do not intersect with those from $d^0$.

\begin{figure}[!h]
    \centering
    \begin{subfigure}[t]{\textwidth}
        \centering
\begin{center}
\begin{tikzpicture}[scale=.5]
\draw[rounded corners,dashed,color=gray] (-2,-2) rectangle (24,8);
\draw[rounded corners,fill=blue!05] (0,2) rectangle (2,4);
\draw[color=MyBlue] (1,1.5) node{\large $\bf{a}$};
\node[draw,circle, line width=.3mm, minimum size=.7cm] (Pa) at (1,3){};

\draw[rounded corners,fill=blue!05] (4,2) rectangle (6,4);
\draw[color=MyBlue] (5,1.5) node{\large $\bf{b}$};
\node[draw,circle, line width=.3mm, minimum size=.7cm] (Pb) at (5,3){};

\draw[rounded corners,fill=blue!05] (8,4) rectangle (10,6);
\draw[color=MyBlue] (9,6.5) node{\large $\bf{e}$};
\node[draw,circle, line width=.3mm, minimum size=.7cm] (Pc) at (9,5){};

\draw[rounded corners,fill=blue!05] (8,0) rectangle (10,2);
\draw[color=MyBlue] (9,-0.5) node{\large $\bf{c}$};
\node[draw,circle, line width=.3mm, minimum size=.7cm] (Pe) at (9,1){};

\draw[->,>=latex, line width=.3mm] (Pa) to[bend right] (Pb);
\draw[->,>=latex, line width=.3mm] (Pb) to[bend right] (Pa);
\draw[->,>=latex, line width=.3mm] (Pc) to[bend right] (Pb);
\draw[->,>=latex, line width=.3mm] (Pe) to[bend left] (Pb);

\draw[rounded corners,fill=blue!05] (12,0) rectangle (14,6);
\draw[color=MyBlue] (13,-0.5) node{\large $\bf{z}$};
\node[draw,circle, line width=.3mm, minimum size=.7cm] (Pz1) at (13,1){};
\node[draw,circle, line width=.3mm, minimum size=.7cm] (Pz2) at (13,5){};

\draw[rounded corners,fill=blue!05] (16,2) rectangle (18,4);
\draw[color=MyBlue] (17,1.5) node{\large $\bf{y}$};
\node[draw,circle, line width=.3mm, minimum size=.7cm] (Py) at (17,3){};

\draw[rounded corners,fill=blue!05] (20,2) rectangle (22,4);
\draw[color=MyBlue] (21,1.5) node{\large $\bf{x}$};
\node[draw,circle, line width=.3mm, minimum size=.7cm] (Px) at (21,3){};

\draw[->,>=latex, line width=.3mm] (Px) to[bend right] (Py);
\draw[->,>=latex, line width=.3mm] (Py) to[bend right] (Px);
\draw[->,>=latex, line width=.3mm] (Pz1) to[bend right] (Py);
\draw[->,>=latex, line width=.3mm] (Pz2) to[bend left] (Py);
\end{tikzpicture}
\end{center}
        \caption{Database $d$} \label{fProp1.1}
    \end{subfigure}

    \vspace{.5cm}
    \begin{subfigure}[t]{\textwidth}
    \centering
\begin{center}
\begin{tikzpicture}[scale=.5]
\draw[rounded corners,dashed,color=gray] (-2,-2) rectangle (24,8);
\draw[rounded corners,fill=blue!05] (0,2) rectangle (2,4);
\draw[color=MyBlue] (1,1.5) node{\large $\bf{a}$};
\node[draw,circle, line width=.3mm, minimum size=.7cm] (Pa) at (1,3){};

\draw[rounded corners,fill=blue!05] (4,2) rectangle (6,4);
\draw[color=MyBlue] (5,1.5) node{\large $\bf{b}$};
\node[draw,circle, line width=.3mm, minimum size=.7cm] (Pb) at (5,3){};

\draw[rounded corners,fill=blue!05] (8,0) rectangle (10,6);
\draw[color=MyBlue] (9,-0.5) node{\large $\bf{c}$};
\node[draw,circle, line width=.3mm, minimum size=.7cm] (Pc1) at (9,5){};
\node[draw,circle, line width=.3mm, minimum size=.7cm] (Pc2) at (9,1){};

\draw[->,>=latex, line width=.3mm] (Pa) to[bend right] (Pb);
\draw[->,>=latex, line width=.3mm] (Pb) to[bend right] (Pa);
\draw[->,>=latex, line width=.3mm] (Pc1) to[bend right] (Pb);
\draw[->,>=latex, line width=.3mm] (Pc2) to[bend left] (Pb);

\draw[rounded corners,fill=blue!05] (12,4) rectangle (14,6);
\draw[color=MyBlue] (13,6.5) node{\large $\bf{e}$};
\node[draw,circle, line width=.3mm, minimum size=.7cm] (Pe) at (13,5){};

\draw[rounded corners,fill=blue!05] (12,0) rectangle (14,2);
\draw[color=MyBlue] (13,-0.5) node{\large $\bf{z}$};
\node[draw,circle, line width=.3mm, minimum size=.7cm] (Pz) at (13,1){};

\draw[rounded corners,fill=blue!05] (16,2) rectangle (18,4);
\draw[color=MyBlue] (17,1.5) node{\large $\bf{y}$};
\node[draw,circle, line width=.3mm, minimum size=.7cm] (Py) at (17,3){};

\draw[rounded corners,fill=blue!05] (20,2) rectangle (22,4);
\draw[color=MyBlue] (21,1.5) node{\large $\bf{x}$};
\node[draw,circle, line width=.3mm, minimum size=.7cm] (Px) at (21,3){};

\draw[->,>=latex, line width=.3mm] (Px) to[bend right] (Py);
\draw[->,>=latex, line width=.3mm] (Py) to[bend right] (Px);
\draw[->,>=latex, line width=.3mm] (Pz) to[bend right] (Py);
\draw[->,>=latex, line width=.3mm] (Pe) to[bend left] (Py);
\end{tikzpicture}
\end{center}
        \caption{Database $d'$} \label{fProp1.2}
    \end{subfigure}
    \caption{Databases in the proof of \propref{pImpossibility}} \label{fProp1}
\end{figure}

\FloatBarrier

By Null Author, we must have $f_c(d)=f_e(d)=f_z(d)=0$. By Field Comparability, we have:
$$\dfrac{f_a(d)+f_b(d)}{4}=\dfrac{f_x(d)+f_y(d)}{3}.$$
Moreover, as $\sum_{a'\in \mathcal{A}}f_{a'}(d)>0$, we must have $f_a(d)+f_b(d)>f_x(d)+f_y(d)>0$. Now, by Author Anonymity, the scores of $a$, $b$, $x$ and $y$ must be the same in the database $d'$ represented in \figref{fProp1.2}, obtained by having $c$ taking a paper from $e$ and $e$ taking a paper from $z$.

By Field Comparability, we must also have:
$$\dfrac{f_a(d)+f_b(d)}{3}=\dfrac{f_x(d)+f_y(d)}{4}.$$
We thus obtain $f_a(d)+f_b(d)<f_x(d)+f_y(d)$, hence  a contradiction with the previous inequality.
\end{proof}

\subsection{Proof of \propref{pPossibility}}

\begin{proof}
Let $f$ be a neutral,  separable and field-comparable influence measure. Let $d^0=(\mathcal{A}^0,\mathcal{P}^0_{\mathcal{A}^0},n^0)\in \mathbb{D}$ be a database and note $\lambda=\frac{1}{A^0}\sum_{a\in \mathcal{A^0}}f_a(d^0)$. Let $d$ be any database in $\mathbb{D}$. By neutrality, we may assume that the authors and papers of $d$ do not intersect those of $d^0$. By Separability and Field Comparability, we obtain:
$$\lambda= \dfrac{1}{A^0}\sum_{a\in \mathcal{A}^0}f_a(d^0) = \dfrac{1}{A^0}\sum_{a\in \mathcal{A}^0}f_a(d^0\oplus d) =  \dfrac{1}{A}\sum_{a\in \mathcal{A}}f_a(d^0\oplus d) = \dfrac{1}{A}\sum_{a\in \mathcal{A}}f_a( d).$$
We have shown that there exists $\lambda\geq 0$ such that, for any $d\in\mathbb{D}$, $\sum_{a\in \mathcal{A}}f_a( d)= \lambda A$. That is, the measure $f$ satisfies equation \eqref{eAccounting} up to a multiplicative constant.
\end{proof}

\end{document}